\documentclass[11pt]{article}

\usepackage{amssymb}
\usepackage{amsmath}
\usepackage{amsthm}
\usepackage{fancyhdr}
\usepackage{graphics}
\usepackage{graphicx}
\usepackage{environ}
\usepackage{clrscode4e}
\usepackage{framed}
\usepackage{url}
\usepackage[margin=1.0in]{geometry}

\newcommand{\PSPACE}{\mathrm{PSPACE}}

\newcommand{\NL}{\mathrm{NL}}
\newcommand{\ZPP}{\mathrm{ZPP}}
\newcommand{\AM}{\mathrm{AM}}
\newcommand{\MA}{\mathrm{MA}}
\newcommand{\IP}{\mathrm{IP}}
\newcommand{\NP}{\mathrm{NP}}
\newcommand{\BPP}{\mathrm{BPP}}
\newcommand{\RNC}{\mathrm{RNC}}
\newcommand{\coNP}{\mathrm{coNP}}
\newcommand{\coNL}{\mathrm{coNL}}
\newcommand{\coMA}{\mathrm{coMA}}
\newcommand{\coAM}{\mathrm{coAM}}
\newcommand{\coIP}{\mathrm{coIP}}

\newcommand{\psd}{\mathrm{psd}}

\theoremstyle{plain}
\newtheorem{lemma}{Lemma}[section]
\newtheorem*{lemma*}{Lemma}
\newtheorem{theorem}[lemma]{Theorem}
\newtheorem{corollary}[lemma]{Corollary}
\newtheorem*{corollary*}{Corollary}

\theoremstyle{definition}
\newtheorem{definition}[lemma]{Definition}

\author{
	Shafi Goldwasser
	\and
	Ofer Grossman
	\and
	Dhiraj Holden
}

\begin{document}

\title{Pseudo-Deterministic Proofs}

\maketitle

\begin{abstract}

We introduce {\it pseudo-deterministic interactive proofs} ($\psd\AM$): interactive proof systems for search problems where
the verifier is guaranteed with high probability to output the same output on different executions.
As in the case with classical interactive proofs,
the verifier is a probabilistic polynomial time algorithm interacting with an untrusted powerful prover.

We view pseudo-deterministic interactive proofs as an extension of the study of pseudo-deterministic randomized polynomial time algorithms:  
 the goal  of the latter  is to {\it  find} canonical solutions to search problems whereas the goal of the former is to 
{\it prove} that a solution to a search problem is canonical to a probabilistic polynomial time verifier. 
Alternatively, one may think of the powerful prover as aiding the probabilistic polynomial time verifier to find canonical solutions to search problems, 
with high probability over the randomness of the verifier.  The challenge is that pseudo-determinism should hold not only with respect to the randomness, but also with respect to the prover: a malicious prover should not be able to cause the verifier to output a solution other than the unique canonical one.

We show the following results:
\begin{itemize}
\item
A natural and illustrative example of a search problem in $\psd\AM$ is the language where given two isomorphic graphs $(G_0,G_1)$, the goal is to find an isomorphism $\phi$ from $G_0$ to $G_1$.
We will show a constant round interactive proof where on every pair of input graphs $(G_0,G_1)$, the verifier with high probability will output 
a unique isomorphism $\phi$ from $G_0$ to $G_1$, although many isomorphisms may exist. 
\item
In contrast, we show that it is unlikely that $\psd\AM$ proofs with constant rounds exist for NP-complete problems by showing that if any $\NP$-complete problem has an $\psd\AM$ protocol where the verifier outputs a unique witness with high probability, then the polynomial hierarchy collapses. 
\item
We show that for every problem in search-$\BPP$, there exists a pseudo-deterministic $\MA$ protocol which succeeds on infinitely many input lengths, where the verifier takes subexponential time. 
\item 
Finally, we consider non-deterministic  log-space NL algorithms with canonical outputs, which we name {\it pseudo-deterministic NL}: on every input, for every non-deterministic choice of the algorithm, either the algorithm rejects or it outputs a canonical unique output.  We show that every search problem in search-$\NL$ (solvable by a nondeterministic log-space algorithm), is in pseudo-deterministic $\NL$.
\item
We show that the class of pseudo-deterministic $\AM$ protocols equals the class of problems solvable by polynomial time search algorithms with oracle access to promise-$\AM \cap \coAM$, where queries to the oracle must be in the promise. We show similar results for pseudo-deterministic $\NP$ and pseudo-deterministic $\MA$.
\end{itemize}
\end{abstract}

\section{Introduction}
In \cite{GG11}, Gat and Goldwasser initiated the study of probabilistic (polynomial-time) search algorithms
that, with high probability, output the same solution on different executions. That is, for all inputs $x$, the randomized algorithm $A$ satisfies $Pr_{r_1, r_2} (A(x, r_1) = A(x, r_2)) \ge 1 - 1/poly(n)$.

Another way of viewing such algorithms is that for a fixed binary relation $R$,
for every $x$ the algorithm associates a canonical solution $s(x)$ satisfying $(x, s(x)) \in R$, and on input $x$ the algorithm outputs $s(x)$ with
overwhelmingly high probability. Algorithms that satisfy this condition are called {\it pseudo-deterministic}, because they essentially offer the same functionality as deterministic algorithms; that is, they produce a 
canonical output for each possible input (except with small error probability)\footnote{In fact, by amplifying the success probability, one can ensure that as black boxes, pseudo-deterministic algorithms are indistinguishable from deterministic algorithms by a polynomial time machine.}. In contrast,
arbitrary probabilistic algorithms that solve search problems may output different solutions when
presented with the same input (but using different internal coin tosses); that is, on input $x$, the output may arbitrarily distributed among all valid solutions for
$x$ (e.g. it may be uniformly distributed).

Several pseudo-deterministic algorithms have been found which improve (sometimes significantly) on
the corresponding best known deterministic algorithm. This is the case for finding quadratic non-residues modulo primes, generators for certain cyclic groups, non-zeros of 
multi-variate polynomials, matchings in bipartite graphs in RNC, and sub-linear algorithms for several problems
\cite{GGR,matching,GG11,roots}.
For other problems, such as finding unique primes of a given length, pseudo-deterministic algorithms remain elusive (for the case of primes, it has been shown that there exists a subexponential time pseudo-deterministic algorithm which works on infinitely many input sizes \cite{OS}). 
Many general questions come to mind, the most significant being: Do polynomial-time pseudo-deterministic algorithms exist for all
search problems that can be solved in probabilistic polynomial time? 

In this work we extend the study of pseudo-determinism in the context of probabilistic algorithms to the context of interactive proofs  and non-determinism.
We view  pseudo-deterministic  interactive proofs as a natural 
extension  pseudo-deterministic randomized polynomial time algorithms:  
 the goal  of the latter  is to {\it  find} canonical solutions to search problems whereas the goal of the former is to 
{\it prove} that a solution to a search problem is canonical to a probabilistic polynomial time verifier.

\subsection{Our Contribution}

Consider the search problem of finding a large clique in a graph.   
A  nondeterministic efficient algorithm for this problem exists: simply guess a set of vertices $C$, confirm in polynomial time
that the set of vertices forms a clique, and either  output $C$  or reject if $C$ is not a clique.
Interestingly, in addition to being nondeterministic, there is another feature of this algorithm; on the same input there may be many possible solutions to the search problem and any one of them may be produced as output. 
Namely, on different executions of the algorithm, on the same input graph $G$,  one execution may guess clique $C$ and another execution may guess clique $C' \neq C$, and  both are valid accepting executions. 
The Satisfiability (SAT) problem  is another  example for which the standard non-deterministic algorithm may more than one output when one exists, i.e., there may be more than one accepting path for the nondeterministic machine and different paths result in different satisfying assignments.

A natural question is whether for each satisfiable formula, there exists a unique canonical satisfying assignment
which can be verified in  polynomial time: that is, can the verifier $V$ be convinced both of the satisfiability of the formula and that the satisfying assignment given to $V$ is canonical?
We note that natural candidates which come to mind, such as the lexicographically first satisfying assignment, are not known to be verifiable in polynomial time (but seem to require the power of $\Sigma_2$ computation).
Furthermore, the work of Hemaspaandra et al \cite{HNOS}  implies the collapse of the polynomial time hierarchy if for every SAT formula there exists a unique  assignment which is polynomial time verifiable in a one round deterministic protocol.

\noindent{\bf Pseudo-determinism for Interactive Proofs and Non-Deterministic Computations:}
In this paper, we consider the setting of interactive proofs for {\it search problems}. 
Such interactive proofs consist of a a pair of interacting algorithms:   a probabilistic polynomial time verifier and 
a computationally unbounded prover which on a common input $x$ engage in rounds of interaction at the end of which
the verifier outputs $y$ - a {\it solution} for the search problem on $x$. 
Analogously to the case of interactive proofs for languages, we require that for every $x$, there exists an honest prover which outputs a correct solution when one exists and for all dishonest provers the probability that the verifier will accept an incorrect solution is small.
We are interested in an additional feature:
the verifier is guaranteed with high probability over its randomness to accept a canonical (unique) output solution and otherwise reject. Importantly, a dishonest prover may not cause the verifier to output a solution other than the canonical (unique) one (except with very low probability).

One may think of the powerful prover as aiding the probabilistic polynomial time verifier to find canonical solutions to search problems, 
with high probability over the randomness of the verifier.  The challenge is that pseudo-determinism should hold not only with respect to the randomness, but also with respect to the prover: a malicious prover should not be able to cause the verifier to output a solution other than the  canonical unique one.
In addition to the intrinsic complexity theoretic interest in this problem,  \textit{consistency} or \textit{predictability} of  
different executions on the same input are natural requirements from protocols.  

To formally study this problem, we define {\it pseudo-deterministic IP} ($\psd\IP$) to be the class of
search problems $R$ (relation on inputs and solutions)  for which  there exists a  probabilistic polynomial time verifier 
 for which for every $x \in R_L$, there is a good powerful prover that will convince the verifier to output with high probability a 
unique witness $s(x)$ (referred to as the ``canonical'' witness) such that
$(x,s(x)) \in R$; and for every $x$ not in $R_L$ (the set of $x$ such that there does not exist a $y$ satisfying $(x, y) \in R$), for all provers  the verifier will reject  with high probability. 
Furthermore, for all provers , the probability that
on $x \in R_L$, the verifier will output any witness $y$ other than the ``canonical'' $s(x)$ is small. 
We let $\psd\AM$ refer to those interactive proofs in which a constant number of rounds is used.

We remark that proving uniqueness of a witness  for any $\NP$ problem can be done in by an interactive proof with a $\PSPACE$ prover --
the prover can convince the verifier that a witness provided  is a lexicographically smallest witness -- using an Arthur-Merlin proof which  takes a {\it polynomial number} of rounds of interaction between prover and verifier following
the celebrated Sum-Check protocol.  
The interesting question to ask is: do {\bf constant-round}  pseudo-deterministic interactive proofs exist for
hard problems in $\NP$ for which many witnesses exist?\\


\noindent{\bf Our Results}

\begin{itemize}
\item
We show that there exists a pseudo-deterministic constant-round Arthur-Merlin protocol for finding an isomorphism between two given graphs. 

Recall that the first protocol showing graph non-isomorphism is in constant round $\IP$ was shown by \cite{GMW} and later shown to be possible using public coins via the general transformation of private to public coins \cite{GS}. Our algorithm finds a unique isomorphism by producing the lexicographically first isomorphism. 
In order to prove that a particular isomorphism between input graph pairs is lexicographically smallest, 
the prover will prove in a sequence of sub-protocols to the verifier that a sequence of graphs suitably defined are non-isomorphic.   In an alternative construction, we exhibit an interactive protocol that computes the automorphism group of a graph in a verifiable fashion.
\item
We prove that if any $\NP$-complete problem has a a pseudo-deterministic constant round $\AM$ protocol, then, $\NP \subseteq \coNP/poly$ and the polynomial hierarchy collapses to the third level,
showing that it is unlikely that NP complete problems have pseudo-deterministic constant round $\AM$ protocols.
This result extends the work of \cite{HNOS} which shows that if there are polynomial time unique verifiable proofs for SAT, then the polynomial hierarchy collapses. Essentially, their result held for deterministic interactive proofs (i.e., NP), and we extend their result to probabilistic interactive proofs with constant number of rounds (i.e., AM). 
\item
For every problem in search-$\BPP$, there exists a pseudo-deterministic $\MA$ protocol where the verifier takes subexponential time on infinitely many input lengths.

The idea of the result is to use known circuit lower bounds to get pseudo-deterministic
subexponential time MA protocols for problems in search-BPP for infinitely many input lengths.
We remark that recently Oliveira and Santhanam \cite{OS} showed 
a subexponential time pseudo-deterministic algorithm for infinitely many input lengths for all properties which have inverse polynomial
density, and are testable in polynomial time.
In their construction, the condition of high density is required. In the case of MA, unconditional
circuit lower bounds for $\MA$ with a verifier which runs in exponential time have been shown by Miltersen et al \cite{MVW99}, which means that inverse polynomial density
is no longer required. Hence, we can obtain a pseudo-deterministic MA algorithm  from circuit
lower bounds. Thus, compared to \cite{OS},  our result shows a  pseudo-derandomization (for a subexponential verifier and infinitely many input sizes $n$)
for all problems in search-BPP (and not just those with high density), but requires a prover.
\item
For every search problem in search-$\NL$, there exists a pseudo-deterministic $\NL$ protocol.

We define
{\it pseudo-deterministic NL} to be the class of search problems $R$ (a relation on inputs and solutions) for which there exists log-space  non-deterministic algorithm $M$ (Turing machines) such that for every input $x$, there exists a unique $s(x)$ such that $ R(x,s(x))=1$ and $M(x)$ outputs $s(x)$ or rejects $x$. Namely, there are no two accepting paths for input $x$ that result in different outputs.

\item
We show structural results regarding pseudo-deterministic interactive proofs. Specifically, we show that $\psd\AM$ equals to the class $\mathrm{search-}\mathrm{P}^{\mathrm{promise-}(\AM \cap \coAM)}$, where for valid inputs $x$, all queries to the oracle must be in the promise. We show similar results in the case of pseudo-deterministic $\MA$ and pseudo-deterministic $\NP$.
\end{itemize}

\subsection{Other Related Work}

In their seminal paper on NP with unique solutions, Valiant and Vazirani asked the following question:
is the inherent intractability of NP-complete problems caused by the fact that NP-complete problems have
many solutions? They show this is not the case by exhibiting a problem -- SAT with unique solutions-- which is NP-hard under randomized reductions. They then showed how their result enables to show the NP-hardness
under randomized reductions for a few related problems such as parity-SAT.
We point out that our question is different. We are not restricting our study to problems (e.g.
satisfiable formulas)  with unique solutions. Rather, we consider hard problems for which there may be exponentially many solutions, and ask if one can  focus on one of them and verify it in polynomial time.
In the language of satisfiability, $\phi$ can be any  satisfiable formula with exponentially many satisfying assignments; set $s(\phi)$ to be a unique valued function which outputs a satisfying assignment for $\phi$. We study whether there exists an $s$ which can be efficiently computed, or which has an efficient interactive proof.

The question of computing canonical labellings of graphs was considered by Babai and Luks \cite{BL}  in the early eighties.
Clearly graph isomorphism is polynomial time reducible to computing canonical labellings of graphs (compute the canonical labeling
for your graphs and compare), however it is unknown whether  the two problems are equivalent (although finding
canonical labellings in polynomial time seems to be known for all classes of
graphs for which isomorphism can be computed in polynomial time).  The problem of computing a set of generators (of size $O(\log n)$) of the automorphism
group of a graph $G$  was  shown by Mathon \cite{M79} (among other results) to be polynomial-time reducible to the problem of computing the isomorphism of a graph. We use this in our proof that graph isomorphism is in $\psd\AM$.

Finally, we mention that recently another notion of uniqueness has been studied in the context of
interactive proofs by Reingold et al \cite{RRR},  called  {\it unambiguous interactive proofs} where the
prover has a unique successful strategy. This again differs from pseudo-deterministic interactive proofs, in that we don't assume (nor guarantee)
a unique strategy by the successful
prover, we only require that the prover proves that the solution (or witness) the verifier receives is unique (with high probability).

\begin{section}{Definitions}

In this section, we define pseudo-determinism in the context of nondeterminism and interactive proofs. We begin by defining a search problem. Intuitively speaking, a search problem is a problem where for an input, there may be multiple possible outputs.

\begin{definition}[Search Problem]
A \textit{search problem} is a relation $R$ consisting of pairs $(x, y)$. We define $L_R$ to be the set of $x$'s such that there exists a $y$ satisfying $(x, y) \in R$. An algorithm solving the search problem is an algorithm that, when given $x \in L_R$, finds a $y$ such that $(x, y) \in R$. When $L_R$ contains all strings, we say that $R$ is a \textit{total} search problem. Otherwise, we say $R$ is a \textit{promise} search problem.
\end{definition}

We now define pseudo-determinism in the context of interactive proofs. Intuitively speaking, we say that an interactive proof is pseudo-deterministic if an honest prover causes the verifier to output the same unique solution with high probability, and dishonest provers can only cause the verifier to output either the unique solution or $\bot$ with high probability. In other words, dishonest provers cannot cause the verifier to output an answer which is not the unique answer. We note that we use $\psd\IP$, $\psd\AM$, $\psd\NP$, $\psd\MA$, and so on, to refer to a class of promise problems, unless otherwise stated.

\begin{definition}[Pseudo-deterministic $\IP$]
A search problem $R$ is in \textit{pseudo-deterministic} $\IP$ (often denoted $\psd\IP$) if there exists an interactive protocol between a probabilistic polynomial time verifier algorithm $V$ and
a prover (unbounded algorithm) $P$ such that for every $x \in L_R$, there exists a $s(x)$ satisfying $(x, s(x)) \in R$ and:
\begin{enumerate}
\item
There exists a $P$ such that $\Pr_r[(P,V)(x,r)=s(x)] \ge {2\over 3}$.
\item
For all $P'$, $\Pr_r[(P',V)(x,r)=s(x)$ or $\bot]\ge{2\over 3}$.
\end{enumerate}
And for every $x \notin L_R$, for all provers $P'$, $\Pr_r[(P',V)(x,r) \neq \bot] \le {1\over 3}$.
\end{definition}

One can similarly define pseudo-deterministic $\MA$, and pseudo-deterministic $\AM$, where $\MA$ is a 1-round protocol, and $\AM$ is a 2-round protocol. One can show that any constant-round interactive protocol can be reduced to a 2-round interactive protocol \cite{babai}. Hence, the definition of pseudo-deterministic $\AM$ captures the set of all search problems solvable in a constant number of rounds of interaction. Furthermore, in the definition of pseudo-deterministic $\AM$, we use public coins. One can show that any protocol using private coins can be simulated using public coins\footnote{In the case where the prover is not unbounded, private coins may be more powerful than public coins.}\cite{GS}.

\begin{definition}[Pseudo-deterministic $\AM$]
A search problem $R$ is in \textit{pseudo-deterministic }$\AM$ (often denoted $\psd\AM$) if there exists a probabilistic polynomial time verifier algorithm $V$, polynomials $p$ and $q$, and
for every $x \in L_R$, there exists an $s(x)$ of polynomial size satisfying $(x, s(x)) \in R$ and:
\begin{enumerate}
\item
$\Pr \nolimits _{r\in \{0,1\}^{p(n)}}(\exists z\in \{0,1\}^{q(n)}\,V(x,r,z)=s(x))\geq {2 \over 3}$
\item
$\Pr \nolimits _{r\in \{0,1\}^{p(n)}}(\forall z\in \{0,1\}^{q(n)}\,V(x,r,z) \in \{s(x), \bot\})\geq {2 \over 3}$.
\end{enumerate}
And for every $x \notin L_R$, we have $\Pr \nolimits _{r\in \{0,1\}^{p(n)}}(\forall z\in \{0,1\}^{q(n)}\,V(x,r,z) = \{\bot\})\geq {2 \over 3}$.
\end{definition}
\end{section}

\begin{definition}[Pseudo-deterministic $\MA$]
A search problem $R$ is in  \textit{pseudo-deterministic }$\MA$ if there is a probabilistic polynomial time verifier $V$ such that
for every $x \in L_R$, there exists an $s(x)$ of polynomial size satisfying\footnote{We remark that we use $M$ to denote the proof sent by the prover Merlin, and not the algorithm implemented by the prover.}:
\begin{enumerate}
\item
There exists a message $M$ of polynomial size such that $\Pr_r[V(x,M,r)=s(x)] \ge {2\over 3}$.
\item
For all $M'$, $\Pr_r[V(x,M',r)=s(x)$ or $\bot ]>{2\over 3}$.
\end{enumerate}
And for every $x \notin L_R$, for all $M'$, $\Pr_r[V(x,M',r) \neq \bot] \le {1\over 3}$.
\end{definition}

Pseudo-determinism can similarly be defined in the context of nondeterminism (which can be viewed as a specific case of an interactive proof):

\begin{definition}[Pseudo-deterministic $\NP$]
A search problem $R$ is in  \textit{pseudo-deterministic} $\NP$ if there is a deterministic polynomial time verifier $V$ such that
for every $x \in L_R$, there exists an $s(x)$ of polynomial size satisfying $(x, s(x)) \in R$ and:
\begin{enumerate}
\item
There exists a message $M$ of polynomial size such that $V(x,M)=s(x)$.
\item
For all $M'$, $V(x,M')=s(x)$ or $V(x,M')=\bot$.
\end{enumerate}
And for every $x \notin L_R$, for all $M'$, we have $V(x,M',r) = \bot$.
\end{definition}

One can also consider the following alternate equivalent definition for pseudo-deterministic $\NP$. On an input $x \in L_R$, we require that for all nondeterministic choices of the machine, either the machine outputs $\bot$ or $s(x)$. Furthermore, we require that there exist a sequence of nondeterministic choices such that the output is $s(x)$. In the case where $x \notin L_R$, we require that for all nondeterministic choices of the machine, the output is $\bot$.

Another alternate equivalent definition is that a search problem is in $\NP$ if it can be solved by an $\NP$ machine which has a unique accept configuration (i.e., the set of possible configurations of the $\NP$ machine and its tape when reaching an accept state has a single element).

We now define pseudo-deterministic $\NL$:

\begin{definition}[Pseudo-deterministic $\NL$]
A search problem $R$ is in  \textit{pseudo-deterministic} $\NL$ if there is a nondeterministic log-space machine $V$ such that
for every $x \in L_R$, there exists an $s(x)$ of polynomial size satisfying $(x, s(x)) \in R$ and:
\begin{enumerate}
\item
There exist nondeterministic choices $N$ for the machine such that such that $V(x, N)=s(x)$.
\item
For all possible nondeterministic choices $N'$, $V(x,N')=s(x)$ or $V(x,N')=\bot$.
\end{enumerate}
And for every $x \notin L_R$, for all nondeterministic choices $N'$, $V(x,N') = \bot$.
\end{definition}

It is valuable to contrast the above definition with the definition of search-$\NL$, in which it is possible to have different nondeterministic guesses of the machine result in different answers:
\begin{definition}[search-$\NL$]
A search problem $R$ is in  \textit{search}-$\NL$ if there is a nondeterministic log-space machine $V$ such that
for every $x \in L_R$,
\begin{enumerate}
\item
There exist nondeterministic choices $N$ for the machine such that such that $V(x, N)=y$, and $(x, y) \in R$.
\item
For all possible nondeterministic choices $N'$, $(x, V(x,N')) \in R$, or $V(x,N')=\bot$.
\end{enumerate}
And for every $x \notin L_R$, for all nondeterministic choices $N'$, $V(x,N') = \bot$.
\end{definition}
Intuitively speaking, in the case of search-$\NL$, it is okay for the algorithm to output different correct solutions when using different nondeterministic choices. In the case of pseudo-deterministic-$\NL$, there should not be two nondeterministic choices for the algorithm which result in different answers (on input $x$, every set of nondeterministic choices which leads to an answer which is not $\bot$ should lead to the same answer $s(x)$).

We will also use the following definition of search-$\BPP$, the class of search problem solvable (and verifiable) in probabilistic polynomial time:

\begin{definition}[Search-$\BPP$]
A binary relation $R$ is in \textit{search-}$BPP$ if
\begin{enumerate}
\item
There is a probabilistic polynomial time algorithm $A$ that given $x \in R_L$, outputs a $y$ such that with probability at least $2/3$, $(x, y) \in R$.
\item
There is a probabilistic polynomial time machine $B$ such if $y$ is output by $A$ when run on $x$, and $(x, y) \notin R$, then $B$ rejects on $(x, y)$ with probability at least $2/3$. Furthermore, with probability at least 1/2, $B$ accepts on $(x, y)$ with probability at least 1/2.
\end{enumerate}
When $x \notin R_L$, $A$ outputs $\bot$ with probability at least $2/3$.
\end{definition}

The intuition of the above definition is that $A$ is used to find an output $y$, and then $B$ can be used to verify $y$, and amplify the success probability.

\begin{section}{Pseudo-deterministic-$\AM$ algorithm for graph isomorphism}\label{sec:GI}

In this section, we show an interactive protocol for graph isomorphism that produces a unique isomorphism with high probability. The method we use to do this involves finding the lexicographically first isomorphism using group theory. In particular, the verifier will obtain the automorphism group of one of the graphs from the prover and verify that it is indeed the automorphism group, and then the verifier will convert an isomorphism obtained from the prover into the lexicographically first isomorphism between the two graphs. We will define the group-theoretic terms used below.

We present an alternate proof of the same result in the appendix (the proof in the appendix is more combinatorial, whereas the proof below is more group theoretic).

\begin{definition}[Automorphism Group]
The \textit{automorphism group} $Aut(G)$ of a graph is the set of permutations $\phi : G \rightarrow G$ such that for every $u,v \in V(G)$, $(u,v) \in E(G) \iff (\phi(u),\phi(v)) \in E(G)$ (i.e., $\phi$ is an automorphism of $G$). 
\end{definition}

\begin{definition}[Stabilize]
Given a set $S$ and elements $\alpha_1, \alpha_2,...,\alpha_i \in S$, we say that a permutation $\phi: S \rightarrow S$ \textit{stabilizes} $\{\alpha_1,\alpha_2,...,\alpha_k \}$ iff $\phi(\alpha_i) = \alpha_i$ for $i \in \{1,...,k\}$. We also say that a group $G$ \textit{stabilizes} $\{\alpha_1,\alpha_2,...,\alpha_k \}$ when every $\phi \in G$ stabilizes $\{\alpha_1,\alpha_2,...,\alpha_k \}$.
\end{definition}

\begin{definition}[Stabilizer]
The \textit{stabilizer} of an element $s$ in $S$ for a group $G$ acting on $S$ is the set of elements of $G$ that stabilize $s$.
\end{definition}

\begin{lemma}\label{uniquefromautomorphism}
Suppose that we are given a tuple $(G_1, G_2, H, \phi)$ where $G_1$ and $G_2$ are graphs, $H = Aut(G_1)$ is represented as a set of generators, and $\phi$ an isomorphism between $G_1$ and $G_2$. Then, in polynomial time, we can compute a unique isomorphism $\phi^*$ from $G_1$ to $G_2$ independent of the choice of $\phi$ and the representation of $H$.
\end{lemma}
\begin{proof}
We use the algorithm given in \cite{C73} to compute a canonical coset representative, observing that the set of isomorphisms between $G_1$ and $G_2$ is a coset of the automorphism group of $G_1$. Let $\alpha_1,...,\alpha_t$ be a basis of $H$, i.e., a set such that any $h \in H$ fixing $\alpha_1,...,\alpha_t$ is the identity. Let $H_i$ be the subgroup of $H$ that stabilizes $\alpha_1,...,\alpha_{i-1}$. Now, let $U_i$ be a set of coset representatives of $H_{i+1}$ in $H_i$. Given the generators of $H_i$, we can calculate $U_i$, and by Schreier's theorem we can calculate the generators for $H_{i+1}$. In this fashion, we can get generators and coset representatives for all the $H_i$. To produce $\phi^*$, we do the following. 
\begin{codebox}
\Procname{\textsc{Find-First-Isomorphism}}
\li $\phi^* = \phi$
\li For $i = 1,...,t$ \Do 
    \li Let $P_i = \{ \phi^*u | u \in U_i \}$.
    \li Set $\phi^* = \arg \min_{\phi \in P_i} (\phi(\alpha_i))$.
    \End
\end{codebox}
To see that this produces a unique isomorphism that does not depend on $\phi$, observe that $\phi^*(\alpha_1)$ is the minimum possible value of $\phi(\alpha_1)$ over all isomorphisms of $G_1$ to $G_2$ as $U_1$ is a set of coset representatives for the stabilizer of $\alpha_1$ over $H$. Also, if $\phi^*(\alpha_i)$ is fixed for $i \in \{1,...,k\}$, then $\phi^*(\alpha_{k+1})$ is the minimum possible value of $\phi(\alpha_{k+1})$ over all isomorphisms which take $\alpha_1$ to $\phi^*(\alpha_1)$, $\alpha_2$ to $\phi^*(\alpha_2)$,..., and $\alpha_k$ to $\phi^*(\alpha_k)$, as $U_{i+1}$ stabilizes $\alpha_1,...,\alpha_k$, so everything in $P_{i+1}$ takes $\alpha_1$ to $\phi^*(\alpha_1)$, $\alpha_2$ to $\phi^*(\alpha_2)$,..., and $\alpha_k$ to $\phi^*(\alpha_k)$. This implies that $\phi^*$ does not depend on $\phi$ and is unique.
\end{proof}

Given this result, this means that it suffices to show a protocol that lets the verifier obtain a set of generators for the automorphism group of $G_1$ and an isomorphism that are correct with high probability, as by the above lemma this can be used to obtain a unique isomorphism between $G_1$ and $G_2$ independent of the isomorphism or the generators. 
\begin{theorem}
There exists an interactive protocol for graph isomorphism such that with high probability, the isomorphism that is output by the verifier is unique, where in the case of a cheating prover the verifier fails instead of outputting a non-unique isomorphism. In other words, finding an isomorphism between graphs can be done in $\psd \AM$.
\end{theorem}
\begin{proof}
From Lemma \ref{uniquefromautomorphism}, it suffices to show an interactive protocol that computes the automorphism group of a graph in a verifiable fashion. \cite{M79} reduces the problem of computing the generators of the automorphism group to the problem of finding isomorphisms. Using this reduction, we can make a constant-round interactive protocol to determine the automorphism group by finding the isomorphisms in parallel. The reason we can do this in parallel is that \cite{M79} implies that there are $O(n^4)$ different pairs of graphs to check and for each pair of graphs we either run the graph isomorphism protocol or the graph non-isomorphism protocol. In the case of the graph isomorphism protocol, the verifier need only accept with an isomorphism in hand; for graph non-isomorphism, the messages sent to the prover are indistinguishable between the two graphs when they are isomorphic, so since the graphs and permutations are chosen independently, there is no way for the prover to correlate their answers to gain a higher acceptance probability for isomorphic graphs. Thus this means that the verifier can determine the automorphism group of a graph and verify that it is indeed the entire automorphism group. Using Lemma \ref{uniquefromautomorphism} we then see that the prover just has to give the verifier an isomorphism, and verifier can compute a unique isomorphism using the automorphism group.
\end{proof}

We note that in the above protocol, the prover only needs to have the power to solve graph isomorphism (and graph non-isomorphism). Also, we note that the above protocol uses private coins. While the protocol can be simulated with a public coin protocol \cite{GS}, the simulation requires the prover to be very powerful. It remains open to determine whether there is a pseudo-deterministic $\AM$ protocol for graph isomorphism which uses public coins, and uses a ``weak" prover (one which is a polynomial time machine with access to an oracle solving graph isomorphism).
\end{section}

\begin{section}{Lower bound on pseudo-deterministic $\AM$ algorithms}\label{lowerbound}
In this section, we establish that if any $\NP$-complete problem has an $\AM$ protocol that outputs a unique witness with high probability, then the polynomial hierarchy collapses. We rely on techniques in \cite{HNOS} combined with the fact that $\AM$ is contained in $\NP/poly$. 

We begin by proving that $\psd\AM \subseteq \psd\NP/poly$:

\begin{lemma}\label{np/poly}
Suppose that there is a $\psd\AM$ protocol for a search problem $R$, which on input $x \in L_R$, outputs $f(x)$. Then, the search problem $R$ has a $\psd\NP/poly$ algorithm which, on input $x$, outputs $f(x)$.
\end{lemma}

\begin{proof}
This proof is similar to the proof that $\AM \subseteq \NP/poly$, which uses techniques similar to those of Adleman's theorem \cite{A79}, showing $\BPP \subseteq \mathrm{P}/poly$.

Consider the $\psd\AM$ protocol, and suppose that on input $x \in L_R$, it outputs $f(x)$.

Since we are guaranteed that when the verifier of the the $\psd\AM$ accepts, it will output $f(x)$ with high probability, we can use standard amplification techniques to show that the verifier will output $f(x)$ with probability $1 - o(\exp(-n))$, assuming an honest prover, and will output anything other than $f(x)$ with probability $o(\exp(-n))$, even with a malicious prover. Then, by a union bound, there exists a choice of random string $r$ that makes the verifier output $f(x)$ for all inputs $x \in \{0,1\}^n$ of size $n$ with an honest prover, and that for malicious provers, the verifier will either reject or output $f(x)$. We encode this string $r$ as the advice string for the $\NP/poly$ machine.

The $\NP/poly$ machine computing $f$ can read $r$ off the advice tape and then guess the prover's message, and whenever the verifier accepts, $f(x)$ will be output by that branch. Thus $f(x)$ can be computed by an $\NP/poly$ machine.
\end{proof}

Next, we show that if an $\NP$-complete problem has a pseudo-deterministic-$\NP/poly$ algorithm, then the polynomial hierarchy collapses.

\begin{theorem}
Let $L \in \NP$ be an $\NP$-complete problem. Let $R$ be a polynomial time algorithm such that there exists a polynomial $p$ so that $x \in L$ if and only if $\exists y \in \{0,1\}^{p(|x|)} R(x,y)$. Suppose that there is a $\psd\AM$ protocol that when given some $x\in L$, outputs a unique $f(x) \in \{0,1\}^{p(|x|)}$ such that $R(x,f(x)) = 1$. Then, $\NP \subseteq \coNP/poly$ and the polynomial hierarchy collapses to the third level.
\end{theorem}

\begin{proof}
From Lemma \ref{np/poly}, we have that there exists $\psd\NP/poly$ algorithm that given $x \in L$, outputs a unique witness $f(x)$ for $x$. Given this function $f$, we construct a function $g(x,y)$ which outputs one of $x, y$ or $\bot$, such that either $g(x, y) = \bot$ or $g(x, y) \in L$ (where the latter holds if at least one of $x$ or $y$ is in $L$). Furthermore, we will ensure that $g$ is computable in $\NP/poly$.

To see this, let $g'(x,y)$ be the function which outputs the set $\{x,y\} \cap L$. We construct $g$ by reducing the language $\{ (x,y) | g'(x,y) \neq \emptyset \}$ (which is in $\NP$, and hence reducible to $L$, since $L$ is $\NP$-complete) to $L$ and computing $f$ to find a unique witness in $\NP/poly$. We then test whether that witness is a witness for $x$ or for $y$. If the witness is for $x$, we set $g(x, y) = x$. Otherwise, if the witness is for $y$, we set $g(x, y) = y$. We view $g$ as a function on the set $\{x,y\}$. That is, we set $g(x, y) = g(y, x)$ (if a function $g$ does not satisfy this property, we can create a $g^*$ satisfying this property by setting $g^*(x, y) = g(min(x, y), max(x, y))$).

Using $g$, we will show how to compute $L$ in $\coNP/poly$. We construct the advice string for $L$ for length $n$ as follows. Start out with $S = \emptyset$. Every iteration, we can find a $y \in \{0,1\}^n \cap L$ such that $g(x,y) = x$ for at least half of the set $\{ x \in \{0,1\}^n \cap L | g(x,s) = s \forall s \in S \}$. Such an $s$ exists since in expectation, when picking a random $s$, half of the $x$'s will satisfy $g(x, s) = x$. If we keep doing this, we get a set $S$ with $|S| \leq poly(n)$ such that for every $x \in L$ of length $n$, there exists an $s \in S$ such that $g(x,s) = x$.

Now, the algorithm to check if $x \notin L$ in $\NP/poly$ is as follows: first, we compute $g(x,s)$ for every $s \in S$, where $S$ is constructed as above and put on the advice tape, which can be done in $\NP/poly$, and check that $g(x,s) = s$ for every $s \in S$ which is possible because $|S|$ is polynomial in $n$. It is clear that this algorithm accepts if $x \notin L$ and rejects if $x \in L$, so therefore $L \in \coNP/poly$, which implies that $\NP \subseteq \coNP/poly$. Furthermore, $\NP \subseteq \coNP/poly$ implies that the polynomial hierarchy collapses to the third level.
\end{proof}

\end{section}

\begin{section}{Pseudo-deterministic derandomization for $\BPP$ in subexponential time $\MA$}

In this section, we will show how to use known circuit lower bounds to get pseudo-deterministic subexponential time $\MA$ protocols for problems in search-$\BPP$ for infinitely many input lengths. In \cite{OS}, it is shown that there is a subexponential time pseudo-deterministic $\ZPP$ algorithm for infinitely many input lengths for all properties which have inverse polynomial density, and are testable in polynomial time\footnote{We believe that their analysis can be improved to get a half-exponential time bounded-error pseudo-deterministic algorithm, which is better than our result for properties of inverse polynomial density.}. A notable example of such a property is primality. So as a corollary, in \cite{OS} it is shown that given some integer $n$, one can find a prime greater than $n$ in pseudo-deterministic subexponential time, for infinitely many input lengths.

In their construction, the condition of high density is required because they show that either there exists a subexponential sized hitting set for infinitely many input lengths which can be used to find strings with the property deterministically, or a complexity collapse happens which implies circuit lower bounds which give pseudo-deterministic algorithms. In the case of $\MA$, unconditional circuit lower bounds for $\MA_{EXP}$ have been shown \cite{MVW99}, which means that inverse polynomial density is no longer required. Hence, we can obtain a pseudo-deterministic $\MA$ algorithm directly from circuit lower bounds. Compared to \cite{OS}, our result manages to show a certain pseudo-derandomization for all problems in search-$\BPP$ (and not just those with high density), but requires a prover.

\begin{theorem}
Given a problem $R$ in search-$\BPP$, it is possible to obtain a pseudo-deterministic $\MA$ algorithm for $R$ where the verifier takes subexponential time for infinitely many input lengths.
\end{theorem}

\begin{proof}
From \cite{MVW99}, we see that $\MA_{EXP} \cap \coMA_{EXP}$ cannot be approximated by half-exponential sized circuits for infinitely many input lengths. From \cite{NW94}, this means that in half-exponential time $\MA$, we can construct a pseudorandom generator with half-exponential stretch which is secure against any given polynomial-size circuit for infinitely many input lengths. If the verifier then runs through each output of the pseudorandom generator on the search-BPP problem and for each possible output tests whether it is a valid output (which can be done in $\BPP$, and hence also in $\MA$). Then, it returns the first such valid output.

This will output the same solution whenever the verifier both gets the correct truth-table for the PRG, and succeeds in testing for each PRG output whether the output it provides is valid. Both of these happen with high probability, and thus this is a pseudo-deterministic subexponential-time $\MA$ algorithm for any problem in search-BPP which succeeds on infinitely many input lengths.
\end{proof}

\end{section}

\section{Uniqueness in $\NL$}

In this section, we prove that every problem in search-$\NL$ can be made pseudo-deterministic:

\begin{theorem}[Pseudo-determinism $\NL$]
Every search problem in search-$\NL$ is in $\psd\NL$.
\end{theorem}

Intuitively speaking, one can think of the search problem as: given a directed graph $G$, and two vertices $s$ and $t$ such that there is a path from $s$ to $t$, find a path from $s$ to $t$. Note that the standard nondeterministic algorithm of simply guessing a path will result in different paths for different nondeterministic guesses. Our goal will be to find a unique path, so that on different nondeterministic choices, we will not end up with a path which is not the unique one.

The idea will be to find the lexicographically first shortest path (i.e, if the min-length path from $s$ to $t$ is of length $d$, we will output the lexicographically first path of length $d$ from $s$ to $t$). To do so, first we will determine the length $d$ of the min-length path from $s$ to $t$. Then, for each neighbor of $s$, we will check if it has a path of length $d-1$ to $t$, and move to the first such neighbor. Now, we have reduced the problem to finding a unique path of length $d-1$, which we can do recursively.

The full proof is given below:

\begin{proof}
Given a problem in search-$\NL$, consider the set of all min-length computation histories. We will find the lexicographically first successful computation history in this set.

To do so, we first (nondeterministically) compute the length of the min-length computation history. This can be done because $\coNL = \NL$ (so if the shortest computation history is of size $T$, one can show a history of size $T$. Also, because it is $\coNL$ to show that there is no history of size up to $T-1$, we can show that there is no history of size less than $T$ in $\NL$).

In general, using the same technique, given a state $S$ of the $\NL$ machine, we can tell what is the shortest possible length for a successful computation history starting at $S$.

Our algorithm will proceed as follows. Given a state $S$ (which we initially set to be the initial configuration of the $\NL$ machine), we will compute $T$, the length of the shortest successful computation path starting at $S$. Then, for each possible nondeterministic choice, we will check (in $\NL$) whether there exists a computation history of length $T-1$ given that nondeterministic choice. Then, we will choose the lexicographically first such nondeterministic choice, and recurse.

This algorithm finds the lexicographically first computation path of minimal length which is unique. Hence, the algorithm will always output the same solution (or reject), so the algorithm is pseudo-deterministic.
\end{proof}

\section{Structural Results}

In \cite{GGR}, Goldreich et al showed that the set of total search problems solved by pseudo-deterministic polynomial time randomized algorithms equals the set of total search problems solved by deterministic polynomial time algorithms, with access to an oracle to decision problems in $\BPP$. In \cite{matching}, this result was extended to the context of $\RNC$. We show analogous theorems here. In the context of $\MA$, we show that for total search problems, $\psd\MA = \mathrm{search-}\mathrm{P}^{\MA \cap \coMA}$.\footnote{What we call $\mathrm{search-}\mathrm{P}$ is often denoted as $\mathrm{FP}$.} In other words, any pseudo-deterministic $\MA$ algorithm can be simulated by a polynomial time search algorithm with an oracle solving decision problems in $\MA \cap \coMA$, and vice versa.

In the case of search problems that are not total, we show that $\psd\MA$ equals to the class $\mathrm{search-}\mathrm{P}^{\mathrm{promise-}(\MA \cap \coMA)}$, where when the input $x$ is in $L_R$, all queries to the oracle must be in the promise. We note that generally, when having an oracle to a promise problem, one is allowed to query the oracle on inputs not in the promise, as long as the output of the algorithm as a whole is correct for all possible answers the oracle gives to such queries. In our case, we simply do not allow queries  to the oracle to be in the promise. Such reductions have been called \textit{smart} reductions \cite{smartreductions}.

We show similar theorems for $\AM$, and $\NP$. Specifically, we show $\psd\AM = \mathrm{search-}\mathrm{P}^{\mathrm{promise-}(\AM \cap \coAM)}$ and $\psd\NP = \mathrm{search-}\mathrm{P}^{\mathrm{promise-}(\NP \cap \coNP)}$, where the reductions to the oracles are smart reductions.

In the case of total problems, one can use a similar technique to show $\psd\AM = \mathrm{search-}\mathrm{P}^{\AM \cap \coAM}$ and $\psd\NP = \mathrm{search-}\mathrm{P}^{\NP \cap \coNP}$, where the oracles can only return answers to total decision problems.

\begin{theorem}\label{MAstructure}
The class $\psd\MA$ equals the class $\mathrm{search-}\mathrm{P}^{\mathrm{promise-}(\MA \cap \coMA)}$, where on any input $x \in L_R$, the all queries to the oracle are in the promise.
\end{theorem}
\begin{proof}
The proof is similar to the proofs in \cite{GGR} and \cite{matching} which show similar reductions to decision problems in the context of pseudo-deterministic polynomial time algorithms and pseudo-deterministic NC algorithms.

First, we show that a polynomial time algorithm with an oracle for $\mathrm{promise-}(\MA \cap \coMA)$ decision problems which only asks queries in the promise has a corresponding pseudo-deterministic $\MA$ algorithm. Consider a polynomial time algorithm $A$ which uses an oracle for $\mathrm{promise-}(\MA \cap \coMA)$. We can simulate $A$ by an $\MA$ protocol where the prover sends the verifier the proof for every question which $A$ asks the oracle. Then, the verifier can simply run the algorithm from $A$, and whenever he accesses the oracle, he instead verifies the proof sent to him by the prover.

We note that the condition of a smart reduction is required in order for the prover to be able to send to the verifier the list of all queries $A$ will make to the oracle. If $A$ can ask the oracle queries not in the promise, it may be that on different executions of $A$, different queries will be made to the oracle (since $A$ is a adaptive, and the queries $A$ makes may depend on the answers returned by the oracle for queries not in the promise), so the prover is unable to predict what queries $A$ will need answered.

We now show that a pseudo-deterministic $\MA$ algorithm $B$ has a corresponding polynomial time algorithm $A$ that uses a $\mathrm{promise-}(\MA \cap \coMA)$ oracle while only querying on inputs in the promise. On input $x \in L_R$, the polynomial time algorithm can ask the $\mathrm{promise-}(\MA \cap \coMA)$ oracle for the first bit of the unique answer given by $B$. This is a decision problem in $\mathrm{promise-}(\MA \cap \coMA)$ since it has a constant round interactive proof (namely, run $B$ and then output the first bit). Similarly, the algorithm $A$ can figure out every other bit of the unique answer, and then concatenate those bits to obtain the full output.

Note that it is required that the oracle is for $\mathrm{promise-}(\MA \cap \coMA)$, and not just for $\mathrm{promise-}\MA$, since if one of the bits of the output is 0, the verifier must be able to convince the prover of that (and this would require a $\mathrm{promise-}\coMA$ protocol).
\end{proof}

A very similar proof shows the following: 
\begin{theorem}
The class $\psd\NP$ equals the class $\mathrm{search-}\mathrm{P}^{\mathrm{promise-}(\NP \cap \coNP)}$, where on any input $x \in L_R$, all queries to the oracle are in the promise.
\end{theorem}

We now prove a similar theorem for the case of $\AM$ protocols. We note that this is slightly more subtle, since it's not clear how to simulate a $\mathrm{search-}\mathrm{P}^{\mathrm{promise-}(\AM \cap \coAM)}$ protocol using only a constant number of rounds of interaction, since the search-P algorithm may ask polynomial many queries in an adaptive fashion.
\
\begin{theorem}\label{AMstructure}
The class $\psd\AM$ equals the class $\mathrm{search-}\mathrm{P}^{\mathrm{promise-}(\AM \cap \coAM)}$, where on any input $x \in L_R$, the all queries to the oracle are in the promise.
\end{theorem}
\begin{proof}
First, we show that a polynomial time algorithm with an oracle for $\mathrm{promise-}(\AM \cap \coAM)$ decision problems where the queries are all in the promise has a corresponding pseudo-deterministic $\AM$ algorithm. We proceed similarly to the proof of Theorem \ref{MAstructure}. Consider a polynomial time algorithm $A$ which uses an oracle for $\mathrm{promise-}(\AM \cap \coAM)$. The prover will internally simulate that algorithm $A$, and then send to the verifier a list of all queries that $A$ makes to the $\mathrm{promise-}(\AM \cap \coAM)$ oracle. Then, the prover can prove the answer (in parallel), to all of those queries.

To prove correctness, suppose that the prover lies about at least one of the oracle queries. Then, consider the first oracle query to which the prover lied. Then, by a standard simulation argument, one can show that it can be made overwhelmingly likely that the verifier will discover that the prover lied on that query.

Once all queries have been answered by the verifier the algorithm $B$ can run like $A$, but instead of querying the oracle, it already knows the answer since the prover has proved it to him.

The proof that a pseudo-deterministic $\MA$ algorithm $B$ has a corresponding polynomial time algorithm $A$ that uses an $\mathrm{promise-}(\AM \cap \coAM)$ oracle is identical to the proof of Theorem \ref{MAstructure}
\end{proof}

As a corollary of the above, we learn that private coins are no more powerful than public coins in the pseudo-deterministic setting:

\begin{corollary}
A pseudo-deterministic constant round interactive proof using private coins can be simulated by a pseudo-deterministic constant round interactive proof using public coins.
\end{corollary}
\begin{proof}
By Theorem \ref{AMstructure}, we can view the algorithm as an algorithm in $\mathrm{search-}\mathrm{P}^{\AM \cap \coAM}$.

By a similar argument to that in Theorem \ref{AMstructure}, one can show that $\psd\IP = \mathrm{search-}\mathrm{P}^{\IP \cap \coIP}$, where in this context $\IP$ refers to \textit{constant} round interactive proofs using \textit{private coins}, and $\AM$ refers to constant round interactive proofs using \textit{public coins}. Since $\mathrm{promise-}(\AM \cap \coAM) = \mathrm{promise-}(\IP \cap \coIP)$, since every constant round \textit{private coin} interactive proof for decision problems can be simulated by a constant round interactive proof using \textit{public coins} \cite{GS}, we have:

$$\psd\AM = \mathrm{search-}\mathrm{P}^{\mathrm{promise-}(\AM \cap \coAM)} = \mathrm{search-}\mathrm{P}^{\mathrm{promise-}(\IP \cap \coIP)} = \psd\IP.$$
\end{proof}

\section{Discussion and Open Problems}

\paragraph{Pseudo-determinism and TFNP:} The class of total search problems solvable by pseudo-deterministic $\NP$ algorithms is a very natural subset of TFNP, the set of all total $\NP$ search problems. It is interesting to understand how the set of total $\psd \NP$ problems fits in TFNP. For example, it is not known whether $\mathrm{TFNP} = \psd \NP$. It would be interesting either to show that every problem in TFNP has a pseudo-deterministic $\NP$ algorithm, or to show that under plausible assumptions there is a problem in TFNP which does not have a pseudo-deterministic $\NP$ algorithm.

Similarly, it is interesting to understand the relationship of $\psd \NP$ to other subclasses of TFNP. For example, one can ask whether every problem in PPAD has a pseudo-deterministic $\NP$ algorithm (i.e., given a game, does there exists a pseudo-deterministic $\NP$ or $\AM$ algorithm which outputs a Nash Equilibrium), or whether under plausible assumptions this is not the case. Similar questions can be asked for CLS, PPP, and so on.

\paragraph{Pseudo-determinism in Lattice problems:} There are several problems in the context of lattices which have $\NP$ (and often also $\NP \cap \coNP$) algorithms \cite{latticeNPcoNP}. Notable examples include gap-SVP and gap-CVP, for certain gap sizes. It would be interesting to show pseudo-deterministic interactive proofs for those problems. In other words, one could ask: does there exists an $\AM$ protocol for gap-SVP so that when a short vector exists, the \textit{same} short vector is output every time. Perhaps more interesting would be to show, under plausible cryptographic assumptions, that certain such problems \textit{do not} have $\psd\AM$ protocols.

\paragraph{Pseudo-determinism and Number Theoretic Problems:} The problem of generating primes (given a number $n$, output a prime greater than $n$), and the problem of finding primitive roots (given a prime $p$, find a primitive root mod $p$) have efficient randomized algorithms, and have been studied in the context of pseudo-determinism \cite{roots, GG11, OS}, though no polynomial time pseudo-deterministic algorithms have been found. It is interesting to ask whether these problems have polynomial time $\psd\AM$ protocols.

\paragraph{The Relationship between $\psd\AM$ and $\mathrm{search-}\BPP$:} One of the main open problems in pseudo-determinism is to determine whether every problem in $\mathrm{search-}\BPP$ also has a polynomial time pseudo-deterministic algorithm. This remains unsolved. As a step in that direction (and as an interesting problem on its own), it is interesting to determine whether $\mathrm{search-}\BPP \subseteq \psd\AM$. In this paper, we proved a partial result in this direction, namely that $\mathrm{search-}\BPP \subseteq i.o. \psd\MA_{\mathrm{SUBEXP}}$.

\paragraph{Zero Knowledge Proofs of Uniqueness:} The definition of pseudo-deterministic interactive proofs can be extended to the context of Zero Knowledge. In other words, the verifier gets no information other than the answer, and knowing that it is the unique/canonical answer. It is interesting to examine this notion and understand its relationship to $\psd\AM$.

\paragraph{The Power of the Prover in pseudo-deterministic interactive proofs:} Consider a search problem which can be solved in $\IP$ where the prover, instead of being all-powerful, is computationally limited. We know that such a problem can be solved in $\psd\IP$ if the prover has unlimited computational power (in fact, one can show it is enough for the prover to be in PSPACE). In general, if the prover can be computationally limited for some $\IP$ protocol, can it also be computationally limited for a $\psd\IP$ protocol for the same problem? It is also interesting in general to compare the power needed for the $\psd\IP$ protocol compared to the power needed to solve the search problem non-pseudo-deterministically. Similar questions can be asked in the context of $\AM$.

\paragraph{The Power of the Prover in pseudo-deterministic private vs public coins proofs:} In our $\psd\AM$ protocol for Graph Isomorphism, the verifier uses private coins, and the prover is weak (it can be simulated by a polynomial time machine with an oracle for graph isomorphism). If using public coins, what power would the prover need? In general, it is interesting to compare the power needed by the prover when using private coins vs public coins in $\psd\AM$ and $\psd\IP$ protocols. 

\paragraph{Pseudo-deterministic interactive proofs for setting cryptographic global system parameters:} Suppose an authority must come up with global parameters for a cryptographic protocol (for instance, a prime $p$ and a primitive root $g$ of $p$, which would be needed for a Diffie-Hellman key exchange). It may be important that other parties in the protocol know that the authority did not come up with these parameters because he happens to have a trapdoor to them. If the authority proves to the other parties that the parameters chosen are canonical, the other parties now know that the authority did not just pick these parameters because of a trapdoor (instead, the authority had to pick those parameters, since those are the canonical ones). It would be interesting to come up with a specific example of a protocol along with global parameters for which there is a pseudo-deterministic interactive proof showing the parameters are unique.

\bibliographystyle{plain}
\bibliography{bibfile}

\begin{appendix}
\begin{section}{Alternate Algorithm for Graph Isomorphism in pseudo-deterministic $\AM$}
In this section, we present another $\psd \AM$ algorithm for Graph Isomorphism, this one more combinatorial (as opposed to the more group theoretic approach of the algorithm in Section \ref{sec:GI}).

\begin{proof}
Let the vertices of $G_1$ be $v_1, v_2, \ldots, v_n$, and the vertices of $G_2$ be $u_1, u_2, \ldots, u_n$. We will show an $\AM$ algorithm which outputs a unique isomorphism $\phi$. Our algorithm will proceed in $n$ stages (which we will later show can be parallelized). After the $k$th stage, the values $\phi(v_1), \phi(v_2), \ldots, \phi(v_k)$ will be determined. 

Suppose that the values $\phi(v_1), \phi(v_2), \ldots, \phi(v_k)$ have been determined. Then we will determine the smallest $r$ such that there exists an isomorphism $\phi^*$ such that for $1 \le i \le k$, we have  $\phi^*(v_i) = \phi(v_i)$, and in addition, $\phi^*(v_{k+1}) = u_r$. If we find $r$, we can set $\phi(v_{k+1}) = \phi^*(v_{k+1})$ and continue to the $k+1^{th}$ stage.

To find the correct value of $r$, the (honest) prover will tell the verifier the value of $r$ and $\phi$. Then, to show that the prover is not lying, for each $r' < r$, the prover will prove that there exists no isomorphism $\phi'$ such that for $1 \le i \le k$, we have  $\phi'(v_i) = \phi(v_i)$, and in addition, $\phi'(v_{k+1}) = u_{r'}$. To prove this, the verifier will pick $G_1$ or $G_2$, each with probability $1/2$. If the verifier picked $G_1$, he will randomly shuffle the vertices $v_{k+2}, \ldots, v_{n}$, and send the shuffled graph to the prover. If the verifier picked $G_2$, he will set $u'_i = \phi(v_i)$ for $1 \le i \le k$, and $u'_{k+1} = u_{r'}$, and shuffle the rest of the vertices. If the prover can distinguish between whether the verifier initially picked $G_1$ or $G_2$, then that implies there is no isomorphism sending $v_i$ to $\phi(v_i)$ for $1 \le i \le k$, and sending $v_{k+1}$ to $u_{r'}$. The prover now can show this for all $r' \le r$ (in parallel), as well as exhibit the isomorphism $\phi$, thus proving that $r$ is the minimum value such that there is an isomorphism sending $v_i$ to $\phi(v_i)$ for $1 \le i \le k$, and sending $v_{k+1}$ to $u_{r}$.

We now show that the above $n$ stages can be done in parallel in order to achieve a constant round protocol. To do so, in the first stage, the prover sends the isomorphism $\phi$ to the verifier. Then, the verifier can test (in parallel) for each $k$ whether under the assumption that $\phi(v_1), \phi(v_2), \ldots, \phi(v_k)$ are correct, $\phi(v_{k+1})$ is the lexicographically minimal vertex which $v_{k+1}$ can be sent to. We now show that using this parallelized protocol, the prover cannot cheat. To show this, suppose that the prover sent some $\phi' \neq \phi$. Then, consider the smallest $i$ for which $\phi(v_i) \neq \phi'(v_i)$. The prover will have to prove that $\phi'(v_i)$ is the lexicographically minimal vertex which $v_i$ can be mapped to, given $\phi(v_1), \phi(v_2), \ldots, \phi(v_{i-1})$. We note that by a standard simulation argument, because the questions asked by the verifier in other stages (which are now in parallel) can be simulated by the prover, the above $n$ stages can be done in parallel while maintaining the low failure probability, so the protocol can be adapted so it requires only a constant number of rounds.
\end{proof}
\end{section}
\end{appendix}
\end{document}